\documentclass[12pt,a4paper]{article}
\usepackage{latexsym}
\usepackage{setspace}
\usepackage{ulem}
\usepackage{amsmath,amsfonts,amsthm,amssymb,graphicx,hyperref}
\usepackage{xy,rotating}
\usepackage{pstricks}
\usepackage{mathrsfs}
\usepackage{wrapfig}
\usepackage{graphicx}
\usepackage{epsfig}
\usepackage{colortbl}
\usepackage{geometry}
\usepackage{eqnarray}
\usepackage{algorithm,algorithmic}

\usepackage{mathrsfs}
\usepackage{wrapfig}

\newcommand{\OPT}{\mathrm{OPT}}
\newcommand{\TC}{\mathrm{TC}}
\newcommand{\LP}{\mathrm{LP}}

\theoremstyle{plain}
\newtheorem{thm}{Theorem}
\newtheorem{cor}[thm]{Corollary}
\newtheorem{prop}[thm]{Proposition}
\newtheorem{lem}[thm]{Lemma}

\theoremstyle{definition}

\newcommand{\bd}{C}

\begin{document}
\title{A closest vector problem arising in\\
radiation therapy planning%
\thanks{This research is supported by an ``Actions de Recherche
Concert\'ees'' (ARC) project of the ``Communaut\'e fran\c{c}aise de
Belgique''. C\'eline Engelbeen is a research fellow of the ``Fonds
pour la Formation \`a la Recherche dans l'Industrie et dans
l'Agriculture'' (FRIA) and Antje Kiesel is a research fellow of the German National Academic Foundation.}}

\author{C\'{e}line Engelbeen\thanks{Department of Mathematics,
Universit\'{e} Libre de Bruxelles CP 216, Boulevard du Triomphe,
1050 Brussels, Belgium. E-Mail: \texttt{cengelbe@ulb.ac.be}.}, Samuel Fiorini\thanks{Department of Mathematics,
Universit\'{e} Libre de Bruxelles CP 216, Boulevard du Triomphe,
1050 Brussels, Belgium. E-Mail: \texttt{sfiorini@ulb.ac.be}.}, Antje Kiesel\thanks{Institute for Mathematics, University of Rostock, 18051 Rostock, Germany. E-Mail: \texttt{antje.kiesel@uni-rostock.de}.}}

\maketitle

\begin{abstract}
In this paper we consider the following {\sl closest vector problem}. We are given a set of $0$-$1$ vectors, the {\sl generators}, an integer vector, the {\sl target} vector, and a nonnegative integer $\bd$. Among all vectors that can be written as nonnegative integer linear combinations of the generators, we seek a vector whose $\ell_\infty$-distance to the target vector does not exceed $\bd$, and whose $\ell_1$-distance to the target vector is minimum.

First, we observe that the problem can be solved in polynomial time provided the generators form a totally unimodular matrix. Second, we prove that this problem is NP-hard to approximate within an $O(d)$ additive error, where $d$ denotes the dimension of the ambient vector space. Third, we obtain a polynomial time algorithm that either proves that the given instance has no feasible solution, or returns a vector whose $\ell_\infty$-distance to the target vector is within an $O(\sqrt{d \ln d}\,)$ additive error of $\bd$ and whose $\ell_1$-distance to the target vector is within an $O(d \sqrt{d \ln d}\,)$ additive error of the optimum. This is achieved by randomly rounding an optimal solution to a natural LP relaxation.

The closest vector problem arises in the elaboration of radiation therapy plans. In this context, the target is a nonnegative integer matrix and the generators are certain $0$-$1$ matrices whose rows satisfy the consecutive ones property. Here we begin by considering the version of the problem in which the set of generators comprises all those matrices that have on each nonzero row a number of ones that is at least a certain constant. This set of generators encodes the so-called {\sl minimum separation constraint}. We conclude by giving further results on the approximability of the problem in the context of radiation therapy.

\medskip

\end{abstract}


\section{Introduction}

\begin{figure}[h!]
	\centering
		\includegraphics{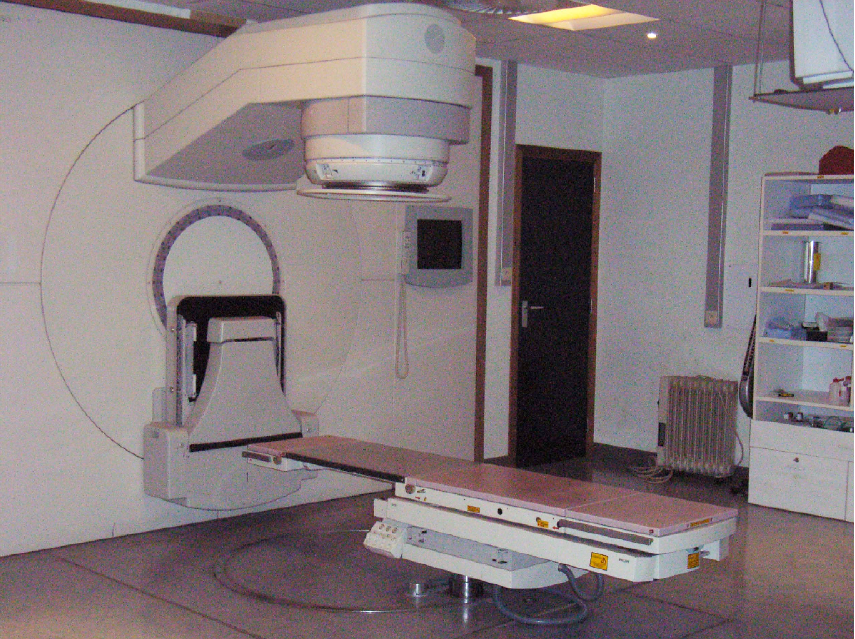}
	\label{fig:linac}
  \caption{A complete treatment unit at Saint-Luc Hospital (Brussels, Belgium).}
\end{figure}

Nowadays, radiation therapy is one of the most used methods for cancer treatment. The aim is to destroy the cancerous tumor by exposing it to radiation while trying to preserve the normal tissues and the healthy organs located in the radiation field. Radiation is commonly delivered by a linear accelerator (see Figure \ref{fig:linac}) whose arm is capable of doing a complete circle around the patient in order to allow different directions of radiation. In intensity modulated radiation therapy (IMRT) a multileaf collimator (MLC, see Figure \ref{fig:mlc}) is used to modulate the radiation beam. This has the effect that radiation can be delivered in a more precise way by forming differently shaped fields and hence improves the quality of the treatment.

\begin{figure}[h!]
	\centering
		\includegraphics[width=5cm]{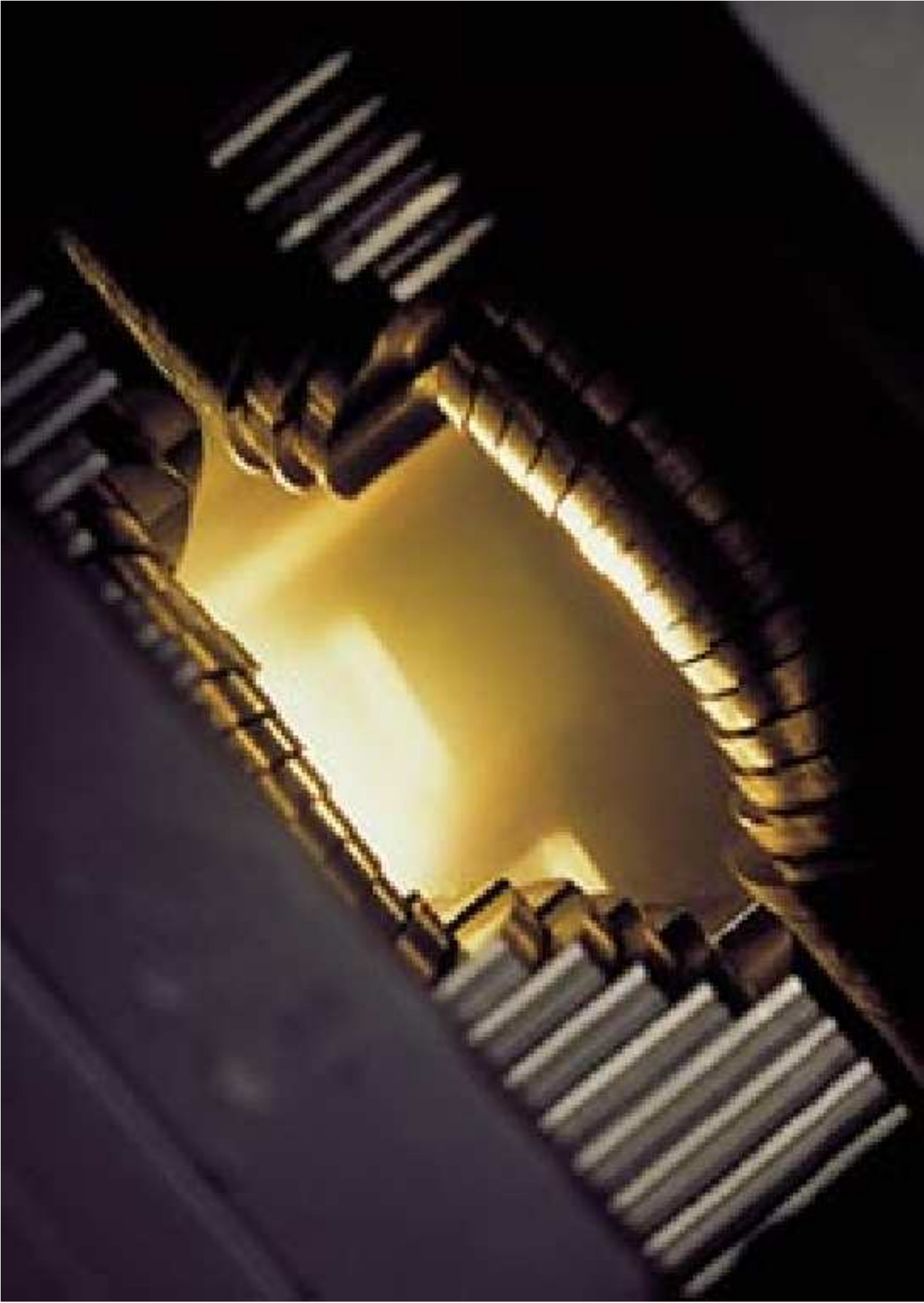}
		\caption{The multileaf collimator (MLC).}
	\label{fig:mlc}
\end{figure}

After the physician has diagnosed the cancer and has located the tumor as well as the organs located in the radiation field, the planning of the radiation therapy sessions is determined in three steps.

In the first step, a number of appropriate beam directions from which radiation will be delivered is chosen \cite{engel_gauer}.

Secondly, the intensity function for each direction is determined. This function is encoded as an integer matrix in which each entry represents an elementary part of the radiation beam (called a {\sl bixel}). The value of each entry is the intensity of the radiation that we want to send through the corresponding bixel.

Finally, the intensity matrix is segmented since the linear accelerator can only send a uniform radiation. This segmentation step mathematically consists in decomposing an $m\times n$ intensity matrix (or {\sl fluence matrix}) $A$ into a nonnegative integer linear combination of certain binary matrices $S=\left(s_{ij}\right)$ that satisfy the {\sl consecutive ones property}. A vector $\mathbf{v} \in \{0,1\}^d$ has the {\sl consecutive ones property}, if $v_{\ell}=1$ and $v_r=1$ for $\ell \leqslant r$ imply $v_j=1$ for all $\ell\leqslant j\leqslant r$. A binary matrix $S$ has the {\sl consecutive ones property}, if each row of $S$ has the consecutive ones property. Such a binary matrix is called a {\sl segment}.

\begin{figure}[ht!]
	\centering
		\includegraphics{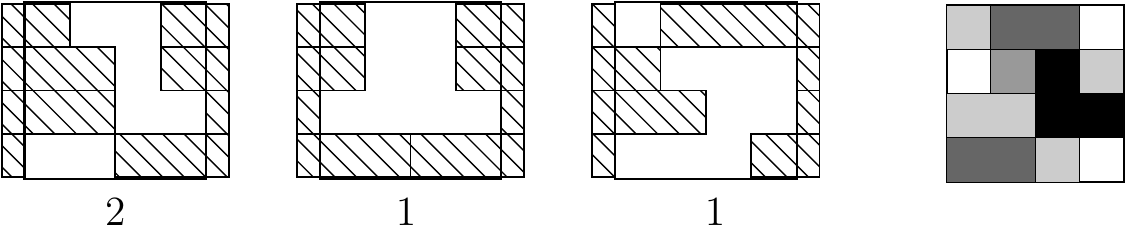}
	\caption{\label{fig:decomp} Leaf positions and irradiation times determining an exact decomposition of a $4 \times 4$ intensity matrix.}
\end{figure}

In this paper, we focus on the case where the MLC is used in the so-called {\sl step-and-shoot} mode, in which the patient is only irradiated when the leaves are not moving. Actually, segments are generated by the MLC and the segmentation step amounts to finding a sequence of MLC positions (see Figure \ref{fig:decomp}). The generated intensity modulated field is just a superposition of homogeneous fields shaped by the MLC.

Throughout the paper, $\left[k\right]$ denotes the set $\left\{1,2,\dots,k\right\}$ for a positive integer $k$, and $\left[\ell,r\right]$ denotes the set $\left\{\ell,\ell+1,\dots,r\right\}$ for positive integers $\ell$ and $r$ with $\ell\leqslant r$. We also allow $\ell=r+1$ where $[\ell,r]= \emptyset$. Thus, an $m \times n$ matrix $S=\left(s_{ij}\right)$ is a segment if and only if there are integer intervals $[\ell_i,r_i]$ for $i \in [m]$ such that
\begin{eqnarray*}
s_{ij}=\begin{cases}
1 & \text{if }j\in[\ell_i,r_i],\\
0 & \text{otherwise.}
\end{cases}
\end{eqnarray*}
A {\sl segmentation} of the intensity matrix $A$ is a decomposition
\begin{eqnarray*}
A=\sum\limits_{j=1}^{k} u_j S_j,
\end{eqnarray*}
where $u_j\in \mathbb{Z}_+$ and $S_j$ is a segment for $j\in \left[k\right]$. The coefficients are required to be integers because in practice the radiation can only be delivered for times that are multiples of a given unit time, called a monitor unit. In clinical applications, a lot of constraints may arise that reduce the number of deliverable segments. For some technical or dosimetric reasons we might look for decompositions where only a subset of all segments is allowed. Those subsets might be explicitely given or defined by constraints like the interleaf collision constraint (denoted by ICC, also called interleaf motion constraint or interdigitation constraint, see \cite{Baatar05}, \cite{Bol04}, \cite{Kal05} and \cite{Kal08}), the interleaf distance constraint (IDC, see \cite{Engelbeen09}), the tongue-and-groove constraint (TGC, see \cite{Bor94}, \cite{Kal08a}, \cite{Kam03a}, \cite{Kam04a}, \cite{Kam04b} and \cite{Que04}), the minimum field size constraint (MFC, see \cite{KiesMFC}), or the minimum separation constraint (MSC, see \cite{Kam03a}).

For some of those constraints, it is still possible to decompose $A$ exactly using the set of feasible segments (like for the ICC, the IDC and the TGC), for others an exact decomposition might be impossible. In this last case, if $\mathcal{S}:=\left\{S_1,\dots,S_k\right\}$ is our set of feasible segments, our aim is to find an approximation $B$ that is decomposable with the segments in $\mathcal{S}$, that satisfies
\begin{equation}
\label{eq:matrix_CVP_constraint}
\| A-B\|_\infty := \max_{i \in [m],\: j \in [n]} |a_{ij}-b_{ij}| \leqslant \bd,
\end{equation}
for some given nonnegative integer constant $\bd$ (possibly, such a matrix $B$ does not exist), and minimizes
\begin{equation}
\label{eq:matrix_CVP_objective}
\| A-B \|_1:=\sum_{i \in [m],\: j \in [n]} |a_{ij}-b_{ij}|.
\end{equation}
The constraint (\ref{eq:matrix_CVP_constraint}) aims at avoiding large bixel-wise differences between target fluence $A$ and realized fluence $B$ (that might lead to undesirable hot spots in the treatment), and the objective (\ref{eq:matrix_CVP_objective}) measures the total change in fluence with respect to the intensity matrix.

Later on in this paper, we will focus on the minimum separation constraint, that imposes a minimum leaf opening $\lambda \in \left[n\right]$ in each open row of the irradiation field. More formally, a segment $S$ given by its leaf positions $(\left[\ell_1,r_1\right],\dots,\left[\ell_m,r_m\right])$ satisfies the minimum separation constraint if and only if $r_i \geqslant \ell_i$ implies $r_i-\ell_i \geqslant \lambda-1$ for all $i \in [m]$. This constraint was first introduced by Kamath, Sahni, Li, Palta and Ranka in \cite{Kam03a}, where the problem of determining if it is possible to decompose $A$ or not under the minimum separation constraint was solved. Here we show that the approximation problem under the minimum separation constraint can be solved in polynomial time with a minimum cost flow algorithm.

The approximation problem described above motivates the definition of the following {\sl Closest Vector Problem (CVP)}. Recall that the {\sl $\ell_\infty$- and $\ell_1$-norms\/} of a vector $\mathbf{x} \in \mathbb{R}^d$ are respectively defined by $\left\|\mathbf{x}\right\|_\infty:=\max_{i \in [d]} \left|x_i\right|$ and $\left\|\mathbf{x}\right\|_1:=\sum^d_{i=1}\left|x_i\right|$. We say that $\bf{x}$ is {\sl binary\/} if $x_i\in \left\{0,1\right\}$ for all $i \in [d]$. The CVP is stated as follows:

\begin{description}
\item[Input:] A collection $\mathcal{G} = \{\mathbf{g}_1,\mathbf{g}_2,\ldots,\mathbf{g}_k\}$ of binary vectors in $\{0,1\}^d$ (the {\sl generators\/}), a vector $\mathbf{a}$ in $\mathbb{Z}_+^d$ (the {\sl target\/} vector), and an upper bound $\bd$ in $\mathbb{Z}_+ \cup \{\infty\}$.

\item[Goal:] Among all vectors $\mathbf{b} := \sum_{j=1}^k u_j\mathbf{g}_j$ with $u_j \in \mathbb{Z}_+$ for $j \in [k]$, find one satisfying $\|\mathbf{a} - \mathbf{b}\|_\infty \leqslant \bd$ and furthermore  minimizing $\|\mathbf{a} - \mathbf{b}\|_1$. If all such vectors $\mathbf{b}$ satisfy $\|\mathbf{a} - \mathbf{b}\|_\infty > \bd$, report that the instance is infeasible.

\item[Measure:] The {\sl total change\/} $\TC := \|\mathbf{a} - \mathbf{b}\|_1$.
\end{description}

We remark that the CVP that is the focus of the present paper differs significantly from the intensively studied CVP on a lattice that is used in cryptography (see, for instance, the recent survey by Micciancio and Regev \cite{MR08}).

In order to cope with the NP-hardness of the CVP, we design (polynomial-time, bi-criteria) approximation algorithms. For the version of the CVP studied here it is natural to consider approximation algorithms with {\sl additive\/} approximation guarantees.

We say that a polynomial-time algorithm is {\sl a $(\Delta_\infty,\Delta_1)$-approximation algorithm} for the CVP if it either proves that the given instance has no feasible solution, or returns a vector $\mathbf{b} = \sum_{j=1}^k u_j \mathbf{g}_j$ with $u_j \in \mathbb{Z}_+$ for $j \in [k]$ such that $\|\mathbf{a}-\mathbf{b}\|_\infty \leqslant \bd + \Delta_\infty$ and $\|\mathbf{a}-\mathbf{b}\|_1 \leqslant \OPT + \Delta_1$, where $\OPT$ is the cost of an optimal solution\footnote{If the instance is infeasible, then we let $\OPT = \infty$.}. Notice that we cannot expect such an approximation algorithm to always either prove that the given instance is infeasible or return a feasible solution, because deciding whether an instance is feasible or not is NP-complete (this claim holds even when $\bd$ is a small constant).

The rest of the paper is organized as follows: Section 2 is devoted to general results on the CVP. We start by observing that the particular case where the generators form a totally unimodular matrix is solvable in polynomial time. We also provide a direct reduction to minimum cost flow when the generators have the consecutive ones property. We afterwards show that, when $\mathcal{G}$ is a general set of generators, for all $\varepsilon > 0$, the CVP admits no polynomial-time $(\Delta_\infty,\Delta_1)$-approximation algorithm with $\Delta_1 \leqslant \left(\ln 2-\varepsilon\right)d$, unless P $=$ NP. (This in particular implies that the CVP is NP-hard.) We conclude the section with an analysis of a natural $(\Delta_\infty,\Delta_1)$-approximation algorithm for the problem based on randomized rounding \cite{MNR97}, with $\Delta_\infty = O(\sqrt{d \ln d}\,)$ and $\Delta_1 = O(d\sqrt{d \ln d}\,)$.

In Section 3, we focus on the particular instances of the CVP arising in IMRT, as described above. We first show, using results of Section 2, that the problem can be solved in polynomial time when the set of generators encodes the minimum separation constraint. We conclude the section with a further hardness of approximation result in case $A$ is a $2\times n$ matrix: for some $\varepsilon > 0$, the problem has no polynomial-time $(\Delta_\infty,\Delta_1)$-approximation algorithm with $\Delta_1 \leqslant \varepsilon\,n$, unless P $=$ NP. (Again, this in particular implies that the corresponding restriction of the CVP is NP-hard.)

In Section 4, we generalize our results to the case where one does not only want to minimize the total change, but a combination of the total change and the sum of the coefficients $u_j$ for $j \in \left[k\right]$. In the IMRT context, this sum represents the total time during which the patient is irradiated, called the {\sl beam-on time}. It is desirable to minimize the beam-on time, for instance, in order to reduce the side effects caused by diffusion of the radiation as much as possible.

Finally, in Section 5, we conclude the paper with some open problems.


\section{The closest vector problem}

In this section we consider the CVP in its most general form. We first consider the particular case where the binary matrix formed by the generators is totally unimodular and prove that the CVP is polynomial in this case. We afterwards prove that, for all $\varepsilon > 0$, there exists no polynomial-time $(\Delta_\infty,\Delta_1)$-approximation algorithm for the general case with $\Delta_1 \leqslant \left(\ln2-\varepsilon\right)d$ unless P $=$ NP. We conclude the section by providing a $(\Delta_\infty,\Delta_1)$-approximation algorithm for the CVP, with $\Delta_\infty = O(\sqrt{d\,\ln d}\,)$ and $\Delta_1 = O(d\sqrt{d\,\ln d}\,)$.


\subsection{Polynomial case} \label{sec:TU}

Consider the following natural LP relaxation of the CVP:
\begin{equationarray}{r@{\quad}rcl@{\qquad}l}
\nonumber \textrm{(LP)} &\multicolumn{3}{@{}l}{\min \sum^{d}_{i=1}(\alpha_{i}+\beta_{i})}\\
\label{eq:(1)}
\mbox{s.t.} &\sum_{j=1}^{k}u_jg_{ij} - \alpha_i + \beta_i &=&a_i& \forall i\in\left[d\right]\\
\label{eq:(3)} &\alpha_i &\geqslant& 0 & \forall i\in\left[d\right]\\[1ex]
\label{eq:(4)} &\beta_i &\geqslant& 0 & \forall i\in\left[d\right]\\[1ex]
\label{eq:(3bis)} &\alpha_i &\leqslant& \bd & \forall i\in\left[d\right]\\[1ex]
\label{eq:(4bis)} &\beta_i &\leqslant& \bd & \forall i\in\left[d\right]\\[1ex]
\label{eq:(2)} &u_j&\geqslant& 0 & \forall j\in\left[k\right].
\end{equationarray}
In this relaxation, the vectors $\boldsymbol{\alpha}$ and $\boldsymbol{\beta}$ model the deviation between the vector $\mathbf{b} := \sum_{j=1}^k u_j \mathbf{g}_j$ and the target vector $\mathbf{a}$. In the IMRT context, $\boldsymbol{\alpha}$ and $\boldsymbol{\beta}$ model the positive and negative differences between realized fluence and target fluence. Clearly, an IP formulation of the CVP can be obtained from (LP) by adding the integrality constraints $u_j \in \mathbb{Z}_+$ for $j \in [k]$.

Let $G$ denote the $d \times k$ binary matrix whose columns are $\mathbf{g}_1$, $\mathbf{g}_2$, \ldots, $\mathbf{g}_k$. If $G$ is totally unimodular, then the same holds for the constraint matrix of (LP). Because $\mathbf{a}$ and $\bd$ are integer, any basic feasible solution of (LP) is integer. Thus, solving the CVP amounts to solving (LP) when $G$ is totally unimodular. Hence, we obtain the following easy result.

\begin{thm} \label{thm:CVP_poly}
The CVP restricted to instances such that the generators form a totally unimodular matrix can be solved in polynomial time.
\end{thm}

\subsection{Minimum cost flow problem}

In this section, we assume that the generators satisfy the consecutive ones property. In particular, $G$ is totally unimodular. This case is of special interest, because it corresponds to the one row case of the segmentation problem in the IMRT context. We show that it is not necessary to solve an LP and provide a direct reduction to the minimum cost flow problem.

We begin by appending a row of zeroes to the matrix $G$ and vector $\mathbf{a}$. Similarly, we add an extra row to the vectors $\boldsymbol{\alpha}$ and $\boldsymbol{\beta}$. Thus the matrix and the vectors now have $d+1$ rows. Next, we replace \eqref{eq:(1)} by an equivalent set of equations: We keep the first equation, and replace each other equation by the difference between this equation and the previous one. Because the resulting constraint matrix is the incidence matrix of a network, we conclude that (LP) actually models a minimum cost network flow problem. We give more details below.

We denote the generators by $\mathbf{g}_{\ell,r}$ where $\left[\ell,r\right]$ is the interval of ones of this generator. That is, $\mathbf{g} = \mathbf{g}_{\ell,r}$ if and only if $g_i = 1$ for $i\in [\ell,r]$ and $g_i = 0$ otherwise. Let $\mathcal{I}$ be the set of intervals such that $\mathcal{G}=\left\{\mathbf{g}_{\ell,r}\;\left|\;\left[\ell,r\right]\in \mathcal{I}\right.\right\}$. We assume that there is no generator with an empty interval of ones (that is, $\ell\leqslant r$ always holds). Now, let $D$ be the network whose set of nodes and set of arcs are respectively defined as:
\begin{eqnarray*}
V(D)&:=&[d+1] = \{1,2,\ldots,d+1\}, \quad \mathrm{and}\\
A(D)&:=&\big\{(i,i+1) \ | \ i \in [d] \big\} \cup \big\{(i+1,i) \ | \ i \in [d] \big\} \cup \big\{(\ell,r+1) \ | \ \left[\ell,r\right] \in \mathcal{I} \big\}.
\end{eqnarray*}
Let us notice that parallel arcs can appear when the interval of a generator only contains one element. In such a case, we keep both arcs: the one representing the generator and the other one.

\begin{figure}[ht!]
\centering
  \includegraphics{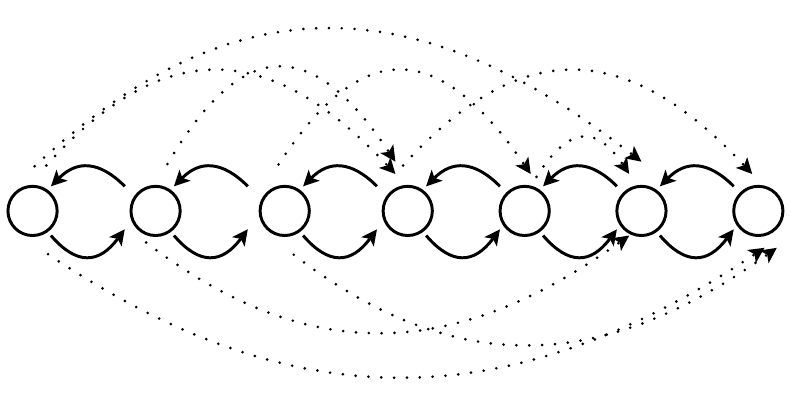}
  \caption{The network for an instance with $d = 6$ and $k = 9$.}
  \label{fig:one_row}
\end{figure}

Letting $a_0 := 0$, we define the demand of each node $j\in V(D)$ as $a_{j-1}-a_{j}$. The arcs of type $(j,j+1)$ and $(j+1,j)$ have capacity $\bd$ and cost $1$. The other arcs, that is, the arcs corresponding to the generators, have infinite capacity and cost $0$. An example of the network is shown in Figure \ref{fig:one_row}. If we consider a flow $\varphi$ in the network, we have the following correspondence between the flow values and the variables of the LP:
\begin{align*}
\varphi(\ell,r+1)&=u_{\ell,r} \text{ for all } [\ell,r] \in \mathcal{I},\\
\varphi(i,i+1)&=\beta_i \text{ for all } i \in [d],\\
\varphi(i+1,i)&=\alpha_i \text{ for all } i \in [d].
\end{align*}
From the discussion above, we obtain the following result.

\begin{prop}\label{poly_CVP_C1}
Let $\mathcal{G}$ and $D$ be as above, let $\mathbf{a} \in \mathbb{Z}_+^d$ and let $\bd \in \mathbb{Z}_+ \cup \{\infty\}$, and let $\OPT$ denote the optimal value of the corresponding CVP instance. Then, $\OPT$ equals the minimum cost of a flow in $D$.
\end{prop}

Our network $D$ is similar to the network used in \cite{ahuja} for finding exact unconstrained decompositions. There, the arcs of type $(i,i+1)$ and $(i+1,i)$ modeling the total change are missing and the arcs of type $(\ell,r+1)$ are available for all nonempty intervals $[\ell,r]$.


\subsection{Hardness}\label{sec:hardness}

In this subsection we prove that the CVP is NP-hard to approximate within an additive error of at most $(\ln 2-\varepsilon)d$, for all $\varepsilon > 0$. To prove this, we consider the particular case where $\mathbf{a}$ is the all-one vector. The given set $\mathcal{G}$ is formed of $k$ binary vectors
$\mathbf{g}_1$, $\mathbf{g}_2$, \dots, $\mathbf{g}_k$. Because $\mathbf{a}$ is binary,
the associated coefficients $u_j$ for $j\in \left[k\right]$ can be assumed to be binary as well. 

For our hardness results, we need a special type of satisfiability problem. A {\sl 3SAT-6 formula} is a conjunctive normal form (CNF formula) in which every clause contains exactly three literals, every literal appears in exactly three clauses and a variable appears at most once in each clause. This means that each variable appears three times negated and three times unnegated. Such a formula is said to be {\sl $\delta$-satisfiable} if at most a $\delta$-fraction of its clauses are satisfiable.

As noted by Feige, Lov\'asz and Tetali \cite{Feige et al}, the following result is a consequence of the PCP theorem (see Arora, Lund, Motwani, Sudan and Szegedy \cite{Arora98}).

\begin{thm}[\cite{Feige et al}]\label{thm:3SAT-6}
There is some $0<\delta<1$, such that it is NP-hard to distinguish between a satisfiable 3SAT-6 formula and one which is $\delta$-satisfiable.
\end{thm}

By combining the above theorem and a reduction due to Feige \cite{Feige} one gets the following result (see Feige, Lov\'asz and Tetali \cite{Feige et al} and also Cardinal, Fiorini and Joret \cite{Sam}).

\begin{lem}[\cite{Sam,Feige et al}]\label{lem:reduction}
For any given constants $c>0$ and $\xi>0$, there is a polynomial time reduction associating to any 3SAT-6 formula $\Phi$ a corresponding set system $\mathbf{\mathcal{S}}(\Phi)=(V,\mathscr{S})$ with the following properties:
\begin{itemize}
\item The sets of $\mathscr{S}$ all have the size $d/t$, where $d=|V|$ and $t$ can be assumed to be arbitrarily large.
\item If $\Phi$ is satisfiable, then $V$ can be covered by $t$ disjoint sets of $\mathscr{S}$.
\item If $\Phi$ is $\delta$-satisfiable, then every $x$ sets chosen from $\mathscr{S}$ cover at most a $1-\left(1-\frac{1}{t}\right)^{x}+\xi$ fraction of the points, for $1\leqslant x\leqslant ct$.
\end{itemize}
\end{lem}

\begin{thm}\label{hard_approx}
For all $\varepsilon > 0$, there exists no polynomial-time $(\Delta_\infty,\Delta_1)$-approximation algorithm for the CVP with $\Delta_1 \leqslant \left(\ln 2-\varepsilon\right)d \approx (0.693 -\varepsilon)d$, unless P $=$ NP.
\end{thm}

\begin{proof}
We define a reduction from 3SAT-6 to the CVP. We use Lemma \ref{lem:reduction} to obtain
a reduction from 3SAT-6 to CVP (by identifying subsets with their characteristic binary vectors) with the following properties. For any given constants $c>0$ and $\xi>0$,
it is possible to set the values of the parameters of the reduction in such a way that:
\begin{itemize}
\item The generators from $\mathcal{G}$ have all the same number $\frac{d}{t}$ of ones, where $t$ can be assumed to be larger than any given constant.
\item If the 3SAT-6 formula $\Phi$ is satisfiable, then $\mathbf{a}$ can be exactly decomposed as a sum of $t$ generators of $\mathcal{G}$.
\item If the 3SAT-6 formula $\Phi$ is $\delta$-satisfiable, then the support of any linear combination of $x$ generators chosen from $\mathcal{G}$ is of size at most $d \left(1-\left(1-\frac{1}{t}\right)^{x}+\xi\right)$, for $1\leqslant x\leqslant ct$.
\end{itemize}

From what precedes, if $\Phi$ is satisfiable then the CVP instance is feasible and $\OPT = 0$. We claim that if $\Phi$ is $\delta$-satisfiable, then \textsl{any} approximation $\mathbf{b} := \sum_{j=1}^k {u_j \mathbf{g}_j}$ with $u_j \in \mathbf{Z}_+$ for $j \in [k]$ has total change $\TC := ||\mathbf{a} - \mathbf{b}||_1 > d\left(\ln 2-\varepsilon\right)$, provided $t$ is large enough and $\xi$ is small enough (this is proved below).

The claim implies the theorem, for the following reason. Assume there exists a polynomial-time $(\Delta_\infty,\Delta_1)$-approximation algorithm with $\Delta_1 \leqslant \left(\ln 2-\varepsilon\right)d$ for the CVP with some nonnegative integer bound $\bd$. Moreover, assume that we are given a 3SAT-6 formula that is either satisfiable or $\delta$-satisfiable.

The approximation algorithm either declares the instance given by the reduction to be infeasible or provides an approximation $\mathbf{b}$. In the first case, we can conclude that $\Phi$ is not satisfiable, hence $\delta$-satisfiable. In the latter case, we compare the total change $\TC$ of the solution returned by the algorithm to $\left(\ln 2-\varepsilon\right)d$. If $\TC \leqslant \left(\ln 2-\varepsilon\right)d$ then the claim implies that $\Phi$ is satisfiable. If $\TC > \left(\ln 2-\varepsilon\right)d$ then we can conclude that $\Phi$ is not satisfiable, hence $\delta$-satisfiable, because otherwise the CVP instance would be feasible with $\OPT = 0$ and the approximation returned by the algorithm should satisfy $\TC \leqslant 0 + \Delta_1 \leqslant \left(\ln 2-\varepsilon\right)d$. In conclusion, we could use the algorithm to decide if $\Phi$ is satisfiable or $\delta$-satisfiable in polynomial time. By Theorem \ref{thm:3SAT-6}, this would imply P $=$ NP, contradiction.

Now, we prove the claim. Notice that we may assume that $u_j \in \{0,1\}$ for all $j \in [k]$. Let $x$ be denote the number of coordinates $u_j$ that are nonzero. We distinguish three cases.
\begin{itemize}
\item \textbf{Case 1: $x=0$.}\medskip

In this case $\TC = d > (\ln 2 - \varepsilon)d$. 
\item \textbf{Case 2 : $1\leqslant x \leqslant ct$.}\medskip

Let $\rho$ denote the number of components $b_i$ of $\mathbf{b}$ that are nonzero. Thus $d-\rho$ is the number of $b_i$ equal to 0. The total change of $\mathbf{b}$ includes one unit for each component of $\mathbf{b}$ that is zero and a certain number of units caused by components of $\mathbf{b}$ larger than one. More precisely, we have:
\begin{eqnarray}
\nonumber\TC&=& d-\rho+x\frac{d}{t}-\rho\\
\nonumber&\geqslant& d\left(\frac{x}{t}+1\right)-2d\left(1-\left(1-\frac{1}{t}\right)^{x}+\xi\right)\\
\nonumber&=&d\left(\frac{x}{t}+2\left(1-\frac{1}{t}\right)^{x}-1-2\xi\right)\\
\nonumber&=&d\left(\left(1-\beta\right)x+2\beta^x-1-2\xi\right),
\end{eqnarray}
where $\beta:=1-\frac{1}{t}$. Note that $\beta<1$ and taking $t$ large corresponds to taking $\beta$ close to 1. In order to derive the desired lower bound on the total change of $\mathbf{b}$ we now study the function $f(x):=\left(1-\beta\right)x+2\beta^x$. The first derivative of $f$ is
\begin{eqnarray*}
f'(x)&=&\left(1-\beta\right)+2\ln\beta\cdot\beta^x.
\end{eqnarray*}
It is easy to verify (since the second derivative of $f$ is always positive) that $f$ is convex and attains its minimum at
\[
x_{\min}=\frac{1}{\ln \beta}\cdot \ln\left(\frac{\beta-1}{2\ln \beta}\right)
\]
Hence we have, for all $x>0$,
\begin{eqnarray*}
f(x)&\geqslant & f(x_{\min})\\
&=&\left(1-\beta\right)x_{\min}+2\beta^{x_{\min}}\\
&=&\left(1-\beta\right)\cdot\frac{1}{\ln \beta}\cdot\ln\left(\frac{\beta-1}{2\ln\beta}\right)+\frac{\beta-1}{\ln \beta}\\
&=&\frac{\beta-1}{\ln \beta}\left(\ln \left(\frac{2\ln \beta}{\beta-1}\right)+1\right).
\end{eqnarray*}
By l'Hospital's rule,
\[
\lim_{\beta\rightarrow 1}\frac{\left(\beta-1\right)}{\ln\beta}=1,
\]
hence we have
\[
f(x)\geqslant\ln 2+1+2\xi-\varepsilon
\]
for $t$ sufficiently large and $\xi$ sufficiently small, which implies
\[
\TC \geqslant d\left(\ln 2-\varepsilon\right).
\]
\item \textbf{Case 3: $x>ct$.}\medskip

Let again $\rho$ be the number of components $b_i$ of $\mathbf{b}$ that are nonzero.
The first $ct$ generators used by the solution have some common nonzero entries. By taking into account the penalties caused by components of $\mathbf{b}$ larger than one, we have:
\begin{eqnarray*}
\TC&\geqslant & ct\cdot\frac{d}{t}-d\left(1-\left(1-\frac{1}{t}\right)^{ct}+\xi\right)\\
&=&d\left(c-1+\left(1-\frac{1}{t}\right)^{ct}-\xi\right)\\
&\geqslant & d\left(\ln2-\varepsilon\right).
\end{eqnarray*}
The last inequality holds for $t$ sufficiently large and $\xi$ sufficiently small and, for instance, $c = 2$.
\end{itemize}

This concludes the proof of the theorem.

\end{proof}


\subsection{Approximation algorithm}

In this subsection we give a polynomial-time $\big(O(\sqrt{d \ln d}\,),O(d\sqrt{d \ln d}\,)\big)$-approximation algorithm for the CVP. This algorithm rounds an optimal solution of the LP relaxation of the CVP given in Section \ref{sec:TU} (see page \pageref{sec:TU}).

If the LP relaxation (LP) is infeasible, then so is the corresponding CVP instance. Now assume that (LP) is feasible and let $\LP$ denote the value of an optimal solution of (LP). Obviously, we have $\OPT \geqslant \LP$.

Note that for each basic feasible solution of (LP), there are at most $d$ components of $\mathbf{u}$ that are nonzero. This is the case, because if we assume that $q > d$ nonzero coefficients exist, then only $k-q$ inequalities of type \eqref{eq:(2)} are satisfied with equality. As we have $2d+k$ variables, we need at least $2d+k$ independent equalities to define a vertex. Thus, there must be $(2d+k)-(k-q)-d=d + q > 2d$ independent inequalities of type \eqref{eq:(3)}, \eqref{eq:(4)}, \eqref{eq:(3bis)} and \eqref{eq:(4bis)} that are satisfied with equality. This is a contradiction, as there can be at most $2d$ such inequalities. Thus, for any extremal optimal solution of the linear program, at most $d$ of the coefficients $u_j$ are nonzero.

\begin{algorithm}[h!]
\caption{~}
\label{rounding-algo}
\begin{algorithmic}
\REQUIRE $\mathbf{a} \in \mathbb{Z}^d_+$, $\bd \in \mathbb{Z}_+ \cup \{\infty\}$, and $\mathbf{g}_1,\mathbf{g}_2,\dots,\mathbf{g}_k \in \left\{0,1\right\}^d$.
\ENSURE An approximation $\mathbf{\tilde{b}}$ of $\mathbf{a}$.
\STATE If (LP) is infeasible, report that the CVP instance is infeasible.
\STATE Otherwise, compute an extremal optimal solution $\left(\boldsymbol{\alpha}^*,\boldsymbol{\beta^*},\mathbf{u}^*\right)$ of (LP).
\FORALL{$j\in\left[k\right]$ }
\STATE if $u^*_j$ is integer $\tilde u_j:=u^*_j$, otherwise
$\tilde u_j:=
\left\{\begin{array}{@{}l@{\ }l}
\lceil u^*_j \rceil & \mbox{with probability } u^*_j-\left\lfloor u^*_j\right\rfloor,\\[1ex]
\lfloor u^*_j\rfloor & \mbox{with probability } \left\lceil u^*_j\right\rceil-u^*_j.
\end{array}\right.$
\ENDFOR
\STATE Return $\mathbf{\tilde{b}}:=\sum_{j=1}^k\tilde u_j\mathbf{g}_{j}$.
\end{algorithmic}
\end{algorithm}

Algorithm \ref{rounding-algo} is an application of the randomized rounding technique. This is a widespread technique for approximating combinatorial optimization problems, see, e.g., the survey by Motwani, Naor and Raghavan~\cite{MNR97}. A basic problem where randomized rounding proves useful is the {\sl lattice approximation problem}: given a binary matrix $H$ of size $d \times d$ and a rational column vector $\mathbf{x} \in [0,1]^d$, find a binary vector $\mathbf{y} \in \{0,1\}^d$ so as to minimize $\left\|H(\mathbf{x}-\mathbf{y})\right\|_{\infty}$.

We will use the following result due to Motwani et al.~\cite{MNR97}. It is a consequence of the Chernoff bound.
\begin{thm}[\cite{MNR97}] \label{thm:Theorem_11.1}
Let $(H,\mathbf{x})$ be an instance of the lattice approximation problem, and let $\mathbf{y}$ be the binary vector obtained by letting $y_j = 1$ with probability $x_j$ and $y_j = 0$ with probability $1-x_j$, independently, for $j \in [d]$. Then the resulting rounded vector $\mathbf{y}$ satisfies $\left\|H(\mathbf{x}-\mathbf{y})\right\|_{\infty} \leqslant \sqrt{4d\ln d}$, with probability at least $1-\frac{1}{d}$.
\end{thm}

We resume our discussion of Algorithm \ref{rounding-algo}. By the discussion above, we know that at most $d$ of the components of $\mathbf{u^*}$ are nonzero. Without loss of generality, we can assume that all nonzero components of $\mathbf{u^*}$ are among its $d$ first components. Then, we let $H$ be the $d\times d$ matrix formed of the first $d$ columns of $G$. (W.l.o.g., we may assume that $d \leqslant k$. If this is not the case we can add generators consisting only of zeros.) Next, we let $\mathbf{x} \in [0,1]^d$ be defined via the following equation (where the floor of the $\mathbf{u^*}$ is computed component-wise):
$$
\mathbf{u}^*-\lfloor \mathbf{u}^*\rfloor =
\left(\begin{array}{c}\mathbf{x}\\ \mathbf{0}\end{array}\right).
$$
Finally, the relationship between the rounded vectors is as follows:
$$
\mathbf{\tilde{u}}-\lfloor \mathbf{u}^*\rfloor =
\left(\begin{array}{c}\mathbf{y}\\ \mathbf{0}\end{array}\right).
$$
We obtain the following result.

\begin{thm} \label{thm:rounding1}
Algorithm \ref{rounding-algo} is a randomized polynomial-time algorithm that either successfully concludes that the given CVP instance is infeasible, or returns a vector $\mathbf{\tilde{b}}$ that is a nonnegative integer linear combination of the generators and satisfies $\|\mathbf{a}-\mathbf{\tilde{b}}\|_{\infty} \leqslant C + \sqrt{4 d \ln d}$ and $\|\mathbf{a}-\mathbf{\tilde{b}}\|_1 \leqslant \OPT + d \sqrt{4 d \ln d}$, with probability at least $1 - \frac{1}{d}$.
\end{thm}
\begin{proof}
Without loss of generality, assume that (LP) is feasible. Thus Algorithm \ref{rounding-algo} returns an approximation $\mathbf{\tilde{b}}$ of $\mathbf{a}$. Let $\mathbf{b^*} = \sum_{j=1}^k u^*_j\,\mathbf{g}_j$. By Theorem \ref{thm:Theorem_11.1} and by the discussion above, we have
$$
\|\mathbf{\tilde{b}} - \mathbf{b^*}\|_\infty=\|G \left(\mathbf{\tilde{u}}-\mathbf{u^*}\right)\|_\infty =\|H \left(\mathbf{x}-\mathbf{y}\right)\|_\infty \leqslant \sqrt{4 d \ln d},
$$
with probability at least $1 - \frac{1}{d}$. Now, the result follows from the inequalities
$$
\|\mathbf{a}-\mathbf{\tilde{b}}\|_\infty
\leqslant \|\mathbf{a}-\mathbf{b^*}\|_\infty +
\|\mathbf{b^*}-\mathbf{\tilde{b}}\|_\infty
\leqslant C + \|\mathbf{b^*}-\mathbf{\tilde{b}}\|_\infty
$$
and
$$
\|\mathbf{a}-\mathbf{\tilde{b}}\|_1
\leqslant \|\mathbf{a}-\mathbf{b^*}\|_1 +
\|\mathbf{b^*}-\mathbf{\tilde{b}}\|_1 \leqslant
\LP + \|\mathbf{b^*}-\mathbf{\tilde{b}}\|_1
\leqslant \OPT + d\,\|\mathbf{b^*}-\mathbf{\tilde{b}}\|_\infty.
$$
\end{proof}

By a result of Raghavan \cite{R88}, Algorithm \ref{rounding-algo} can be derandomized, at the cost of multiplying the additive approximation guarantees $\sqrt{4d\ln d}$ and $d\sqrt{4d\ln d}$ by a constant. We obtain the following result:

\begin{cor}
There exists a polynomial-time $\big(O(\sqrt{d \ln d}\,),O(d\sqrt{d \ln d}\,)\big)$-approximation algorithm for the CVP.
\end{cor}

In the case where $\bd = \infty$, we can slightly improve Theorem \ref{thm:rounding1}, as follows.

\begin{thm} \label{thm:rounding2}
Suppose $\bd = \infty$. Then, Algorithm \ref{rounding-algo} is a randomized polynomial-time algorithm that returns a vector $\mathbf{\tilde{b}}$ that is a nonnegative integer linear combination of the generators and satisfies $\|\mathbf{a}-\mathbf{\tilde{b}}\|_1 \leqslant \OPT + \sqrt{\frac{\ln 2}{2}}\,d \sqrt{d}$ on average.
\end{thm}

Our proof of Theorem \ref{thm:rounding2} uses the following lemma, which is proved in the appendix.

\begin{lem}\label{lem:expectation_abs_val_lambda}
Let $q$ be a positive integer and let $X_1,X_2,\dots,X_q$ be $q$ independent random variables such that, for all $j\in [q]$, $P[X_j = 1-p_j] = p_j$ and $P[X_j = -p_j] = 1-p_j$. Then
$$
E\Big[|X_1 + X_2 + \cdots + X_q|\Big] \leqslant \sqrt{\frac{\ln 2}{2}}\;\sqrt{q}.
$$
\end{lem}

We are now ready to prove the theorem.

\begin{proof}[Proof of Theorem \ref{thm:rounding2}]
We have
\begin{eqnarray*}
E\left[\left\|\mathbf{a}-\mathbf{\tilde b}\right\|_1\right]&\leqslant & E\Big [\Big \|\mathbf{a}-\mathbf{b^*}\Big \|_1\Big ]+E\left[\left\|\mathbf{b^*}-\mathbf{\tilde b}\right\|_1\right]\\
&=&\LP + E\left[\left\|\sum^k_{j=1}u^*_j\mathbf{g}_{j}-\sum_{j=1}^k\tilde{u}_j \mathbf{g}_{j}\right\|_1\right]\\
&=&\LP + E\left[\sum^d_{i=1}\left|\sum^k_{j=1}g_{ij}\left(u^*_j-\tilde{u}_j\right)\right|\right].
\end{eqnarray*}
Without loss of generality, we may assume that $u_j^* = 0$, and thus $\tilde{u}_j = 0$, for $j > d$. This is due to the fact that $\mathbf{u}^*$ is a basic feasible solution, see the above discussion.

Now, let $X_{ij}:=g_{ij}\left(u^*_j-\tilde{u}_j\right)$ for all $i,j \in [d]$. For each fixed $i\in [d]$, $X_{i1},\dots,X_{id}$ are independent random variables satisfying $X_{ij}=0$ if $g_{ij}=0$ or $u^*_{j}\in\mathbb{Z}_+$ and otherwise
\[X_{ij}=
\left\{
\begin{array}{ll}
u^*_j-\left\lceil u^*_j\right\rceil & \text{with probability } u^*_j-\left\lfloor u^*_j\right\rfloor,\\[1ex]
u^*_j-\left\lfloor u^*_j\right\rfloor & \text{with probability } \left\lceil u^*_j\right\rceil-u^*_j.
\end{array}\right.
\]
By Lemma \ref{lem:expectation_abs_val_lambda}, we get:
\begin{eqnarray*}
E\left[\left\|\mathbf{a}-\mathbf{\tilde b}\right\|_1\right]&\leqslant &\LP + E\left[\sum^d_{i=1}\left|X_{i1} + X_{i2} + \cdots + X_{id}\right|\right]\\
&= & \LP + \sum_{i=1}^d E\Big[\left|X_{i1} + X_{i2} + \cdots + X_{id}\right|\Big]\\
&\leqslant &\LP + \sqrt{\frac{\ln 2}{2}}\;d\sqrt{d}\\
&\leqslant &\OPT + \sqrt{\frac{\ln 2}{2}}\;d\sqrt{d}.
\end{eqnarray*}
\end{proof}

A natural question is the following: Is it possible to derandomize Algorithm \ref{rounding-algo} in order to obtain a polynomial-time approximation algorithm for the CVP that provides a total change of at most $\OPT + O(d\sqrt{d}\,)$, provided that $\bd = \infty$? We leave this question open.


\section{Application to IMRT}

In this section we consider a target matrix $A$ and the set of generators $\mathcal{S}$ formed by segments (that is, binary matrices whose rows satisfy the consecutive ones property). In the first part of this section we consider the case where $\mathcal{S}$ is formed by all the segments that satisfy the minimum separation constraint. In the last part we consider any set of segments $\mathcal{S}$. We show that in this last case the problem is hard to approximate, even if the matrix has only two rows.


\subsection{The minimum separation constraint}

In this subsection we consider the CVP under the constraint that the set $\mathcal{S}$ of generators is formed by all segments that satisfy the minimum separation constraint. Given $\lambda\in \left[n\right]$, this constraint requires that the rows which are not totally closed have a leaf opening of at least $\lambda$. Mathematically, the leaf positions of open rows $i \in [m]$ have to satisfy $r_i-\ell_i\geqslant\lambda-1$. We cannot decompose any matrix $A$ under this constraint. Indeed, the following single row matrix cannot be decomposed for $\lambda=3$:
\[A=\left(\begin{array}{ccccc}1 & 1 & 4 & 1 & 1\end{array}\right).\]
The problem of determining if it is possible to decompose a matrix $A$ under this constraint was proved to be polynomial by Kamath et al.\ \cite{Kam03a}.

Obviously, the minimum separation constraint is a restriction on the leaf openings in each single row, but does not affect the combination of leaf openings in different rows. Again, more formally, the set of allowed leaf openings in one row $i$ is
\begin{eqnarray*}
\mathcal{S}_i=\{[\ell_i,r_i] \ | \ r_i-\ell_i \geqslant \lambda-1 \text{ or } r_i=\ell_i-1\},
\end{eqnarray*}
and does not depend on $i$. If we denote a segment by the set of its leaf positions $([\ell_1,r_1],\dots,[\ell_m,r_m])$ then the set of feasible segments $\mathcal{S}$ for the minimum separation constraint is simply $\mathcal{S} = \mathcal{S}_1 \times \mathcal{S}_2 \times \dots \times \mathcal{S}_m$. Thus, in order to solve the CVP under the minimum separation constraint, it is sufficient to focus on single rows.

Indeed, whenever the set of feasible segments has a structure of the form $\mathcal{S}=\mathcal{S}_1 \times \mathcal{S}_2 \times \dots \times \mathcal{S}_m$, which means that the single row solutions can be combined arbitrarily and we always get a feasible segment, solving the single row problem is sufficient.

From Theorem \ref{thm:CVP_poly}, we infer our next result.

\begin{cor}
The restriction of the CVP where the vectors are $m \times n$ matrices and the set of generators is the set of all segments satisfying the minimum separation constraint can be solved in polynomial time.
\end{cor}


\subsection{Further hardness results}

As we have seen in Subsection 2.1, the CVP with generators satisfying the consecutive ones property is polynomial (see Theorem \ref{thm:CVP_poly}). Moreover, we have proved in Theorem \ref{hard_approx} that the CVP is hard to approximate within an additive error of $\left(\ln 2-\varepsilon\right)d$ for a general set of generators. We now prove that, surprisingly, the case where generators contain at most two blocks of ones, which corresponds in the IMRT context of having a $2 \times n$ intensity matrix $A$ and a set of generators formed by $2 \times n$ segments, is NP-hard to approximate within an additive error of $\varepsilon\,n$, for some $\varepsilon>0$.

\begin{thm}
There exists some $\varepsilon>0$ such that the CVP, restricted to $2\times n$ matrices and generators with their ones consecutive on each row, admits no polynomial-time $(\Delta_\infty,\Delta_1)$-approximation algorithm with $\Delta_1 \leqslant \varepsilon\,n$, unless P $=$ NP.
\end{thm}

\begin{proof}
We prove the theorem again by reducing from the promise problem that was introduced in Section \ref{sec:hardness}. Recall that a 3SAT-6 formula is a CNF formula in which each clause contains three literals, and each variable appears non-negated in three clauses, and negated in three other clauses. The problem consists in distinguishing between a formula that is satisfiable and a formula that is $\delta$-satisfiable.

Let a 3SAT-6 formula $\Phi$ in the variables $x_1,\dots,x_s$ be given. Let $c_1$, \ldots, $c_t$ denote the clauses of $\Phi$. We say that a variable $x_i$ and a clause $c_j$ are {\sl incident} if $c_j$ involves $x_i$ or its negation $\bar{x}_i$. By double-counting the incidence between variables and clauses, we find $6s = 3t$, that is,  $t = 2s$.

We build an instance of the restricted CVP as follows: let $A=(a_{ij})$ be the matrix with $2$ rows and $10s$ columns defined as 
\[
a_{ij}=\begin{cases}
0 & \text{if } i=1 \text{ and } 6s+1 \leqslant j \leqslant 10s,\\
1 & \text{otherwise.}\end{cases}
\]
Thus $A$ has $6s$ ones, followed by $4s$ zeros in its first row, and $10s$ ones in its second row. The size of $A$ is $2 \times n$, where $n = 10s=5t$.

To each variable $x_i$ there corresponds an interval $I(x_i) := [6i-5,6i]$ of $6$ ones in the first row of $A$ (in this proof, for the sake of simplicity, we identify intervals of the form $[\ell,r]$ and the $1 \times n$ row vector they represent). We let also $I(\bar{x}_i) := I(x_i)$.

For each interval $I(x_i) = I(\bar{x}_i) = [6i-5,6i]$ we consider two decompositions into three sub-intervals that correspond to setting the variable $x_i$ true or false. The decomposition corresponding to setting $x_i$ true is $I(x_i) = I_1(x_i) + I_2(x_i) + I_3(x_i)$, where $I_1(x_i) :=  [6i-5,6i-5]$, $I_2(x_i) := [6i-4,6i-2]$ and $I_3(x_i) := [6i-1,6i]$. The decomposition corresponding to setting $x_i$ false is $I(\bar{x}_i) = I_1(\bar{x}_i) + I_2(\bar{x}_i) + I_3(\bar{x}_i)$, where $I_1(\bar{x}_i) := [6i-5,6i-4]$, $I_2(\bar{x}_i) := [6i-3,6i-1]$ and $I_3(\bar{x}_i) := [6i,6i]$. An illustration is given in Figure \ref{fig:var_int}.

\begin{figure}[h]
\centering
\unitlength=1.25cm
\begin{tabular}{c@{\hspace{3cm}}c}
\begin{picture}(2.9,1)
\multiput(0.1,.1)(.4,0){7}{\line(0,1){.5}}
\multiput(0.1,.1)(0,.4){2}{\line(1,0){2.4}}
\put(0.1,.1){\rule{.5cm}{.5cm}}
\put(0,.7){\footnotesize $-5$}
\put(0.45,.7){\footnotesize $-4$}
\put(0.875,.7){\footnotesize $-3$}
\put(1.3,.7){\footnotesize $-2$}
\put(1.725,.7){\footnotesize $-1$}
\put(2.25,.7){\footnotesize $0$}
\put(2.65,.2){\footnotesize $I_1(x_i)$}
\end{picture}
&\begin{picture}(2.9,1)
\multiput(0.1,.1)(.4,0){7}{\line(0,1){.5}}
\multiput(0.1,.1)(0,.4){2}{\line(1,0){2.4}}
\put(0.1,.1){\rule{1cm}{.5cm}}
\put(0,.7){\footnotesize $-5$}
\put(0.45,.7){\footnotesize $-4$}
\put(0.875,.7){\footnotesize $-3$}
\put(1.3,.7){\footnotesize $-2$}
\put(1.725,.7){\footnotesize $-1$}
\put(2.25,.7){\footnotesize $0$}
\put(2.65,.2){\footnotesize $I_1(\bar{x}_i)$}
\end{picture}\\
\begin{picture}(2.9,0.75)
\multiput(0.1,.1)(.4,0){7}{\line(0,1){.5}}
\multiput(0.1,.1)(0,.4){2}{\line(1,0){2.4}}
\put(0.5,.1){\rule{1.5cm}{.5cm}}
\put(2.65,.2){\footnotesize $I_2(x_i)$}
\end{picture}
&\begin{picture}(2.9,0.75)
\multiput(0.1,.1)(.4,0){7}{\line(0,1){.5}}
\multiput(0.1,.1)(0,.4){2}{\line(1,0){2.4}}
\put(0.9,.1){\rule{1.5cm}{.5cm}}
\put(2.65,.2){\footnotesize $I_2(\bar{x}_i)$}
\end{picture}\\
\begin{picture}(2.9,0.75)
\multiput(0.1,.1)(.4,0){7}{\line(0,1){.5}}
\multiput(0.1,.1)(0,.4){2}{\line(1,0){2.4}}
\put(1.7,.1){\rule{1cm}{.5cm}}
\put(2.65,.2){\footnotesize $I_3(x_i)$}
\end{picture}
&\begin{picture}(2.9,0.75)
\multiput(0.1,.1)(.4,0){7}{\line(0,1){.5}}
\multiput(0.1,.1)(0,.4){2}{\line(1,0){2.4}}
\put(2.1,.1){\rule{.5cm}{.5cm}}
\put(2.65,.2){\footnotesize $I_3(\bar{x}_i)$}
\end{picture}
\end{tabular}
\caption{\label{fig:var_int} The sub-intervals used for the variables.}
\end{figure}
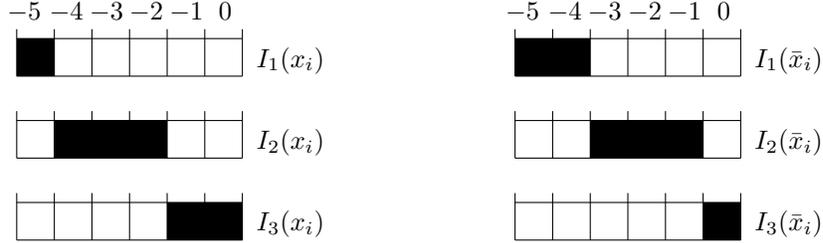

Similarly, to each clause $c_j$ there corresponds an interval $I(c_j) := [5j-4,5j]$ of $5$ ones in the second row of $A$. We define ten sub-intervals that can be combined in several ways to decompose $I(c_j)$. We let $I_1(c_j) := [5j-4,5j-4]$, $I_2(c_j) := [5j-2,5j-2]$, $I_3(c_j) := [5j,5j]$, $I_4(c_j) := [5j-3,5j-3]$, $I_5(c_j) := [5j-1,5j-1]$, $I_6(c_j) := [5j-4,5j-3]$, $I_7(c_j) := [5j-1,5j]$, $I_8(c_j):=[5j-3,5j-1]$, $I_9(c_j) := [5j-4,5j-1]$ and $I_{10}(c_j) := [5j-3,5j]$. An illustration is given in Figure \ref{fig:cl_int}.

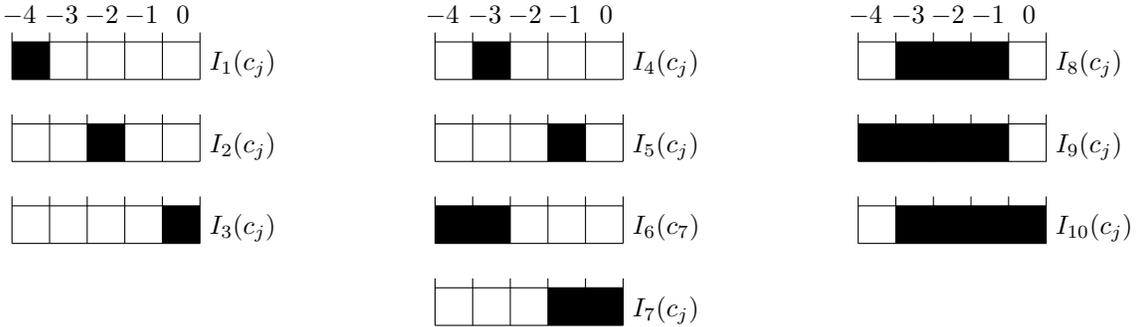
\begin{figure}[h]
\centering
\unitlength=1.25cm
\begin{tabular}{c@{\hspace{2cm}}c@{\hspace{2cm}}c}
\begin{picture}(2.9,1)
\multiput(0.1,.1)(.4,0){6}{\line(0,1){.5}}
\multiput(0.1,.1)(0,.4){2}{\line(1,0){2}}
\put(0.1,.1){\rule{.5cm}{.5cm}}
\put(0,.7){\footnotesize $-4$}
\put(0.45,.7){\footnotesize $-3$}
\put(0.875,.7){\footnotesize $-2$}
\put(1.3,.7){\footnotesize $-1$}
\put(1.85,.7){\footnotesize $0$}
\put(2.2,.2){\footnotesize $I_1(c_j)$}
\end{picture}
&\begin{picture}(2.9,1)
\multiput(0.1,.1)(.4,0){6}{\line(0,1){.5}}
\multiput(0.1,.1)(0,.4){2}{\line(1,0){2}}
\put(0.5,.1){\rule{.5cm}{.5cm}}
\put(0,.7){\footnotesize $-4$}
\put(0.45,.7){\footnotesize $-3$}
\put(0.875,.7){\footnotesize $-2$}
\put(1.3,.7){\footnotesize $-1$}
\put(1.85,.7){\footnotesize $0$}
\put(2.2,.2){\footnotesize $I_4(c_j)$}
\end{picture}
&\begin{picture}(2.9,.75)
\multiput(0.1,.1)(.4,0){6}{\line(0,1){.5}}
\multiput(0.1,.1)(0,.4){2}{\line(1,0){2}}
\put(0.5,.1){\rule{1.5cm}{.5cm}}
\put(0,.7){\footnotesize $-4$}
\put(0.45,.7){\footnotesize $-3$}
\put(0.875,.7){\footnotesize $-2$}
\put(1.3,.7){\footnotesize $-1$}
\put(1.85,.7){\footnotesize $0$}
\put(2.2,.2){\footnotesize $I_8(c_j)$}
\end{picture}\\
\begin{picture}(2.9,0.75)
\multiput(0.1,.1)(.4,0){6}{\line(0,1){.5}}
\multiput(0.1,.1)(0,.4){2}{\line(1,0){2}}
\put(0.9,.1){\rule{.5cm}{.5cm}}
\put(2.2,.2){\footnotesize $I_2(c_j)$}
\end{picture}
&\begin{picture}(2.9,0.75)
\multiput(0.1,.1)(.4,0){6}{\line(0,1){.5}}
\multiput(0.1,.1)(0,.4){2}{\line(1,0){2}}
\put(1.3,.1){\rule{.5cm}{.5cm}}
\put(2.2,.2){\footnotesize $I_5(c_j)$}
\end{picture}
&\begin{picture}(2.9,0.75)
\multiput(0.1,.1)(.4,0){6}{\line(0,1){.5}}
\multiput(0.1,.1)(0,.4){2}{\line(1,0){2}}
\put(0.1,.1){\rule{2cm}{.5cm}}
\put(2.2,.2){\footnotesize $I_9(c_j)$}
\end{picture}\\
\begin{picture}(2.9,0.75)
\multiput(0.1,.1)(.4,0){6}{\line(0,1){.5}}
\multiput(0.1,.1)(0,.4){2}{\line(1,0){2}}
\put(1.7,.1){\rule{.5cm}{.5cm}}
\put(2.2,.2){\footnotesize $I_3(c_j)$}
\end{picture}
&\begin{picture}(2.9,0.75)
\multiput(0.1,.1)(.4,0){6}{\line(0,1){.5}}
\multiput(0.1,.1)(0,.4){2}{\line(1,0){2}}
\put(0.1,.1){\rule{1cm}{.5cm}}
\put(2.2,.2){\footnotesize $I_6(c_7)$}
\end{picture}
&\begin{picture}(2.9,0.75)
\multiput(0.1,.1)(.4,0){6}{\line(0,1){.5}}
\multiput(0.1,.1)(0,.4){2}{\line(1,0){2}}
\put(0.5,.1){\rule{2cm}{.5cm}}
\put(2.2,.2){\footnotesize $I_{10}(c_j)$}
\end{picture}\\
&\begin{picture}(2.9,0.75)
\multiput(0.1,.1)(.4,0){6}{\line(0,1){.5}}
\multiput(0.1,.1)(0,.4){2}{\line(1,0){2}}
\put(1.3,.1){\rule{1cm}{.5cm}}
\put(2.2,.2){\footnotesize $I_7(c_j)$}
\end{picture}
&
\end{tabular}

\caption{\label{fig:cl_int} The sub-intervals used for the clauses.}
\end{figure}

The three first sub-intervals $I_\alpha(c_j)$ ($\alpha \in \{1,2,3\}$) correspond to the three literals of the clause $c_j$. The last seven sub-intervals $I_\alpha(c_j)$ ($\alpha \in \{4,\ldots,10\}$) alone are not sufficient to decompose $I(c_j)$ exactly. In fact, if we prescribe any subset of the first three sub-intervals in a decomposition, we can complete the decomposition using some of the last seven intervals to an exact decomposition of $I(c_j)$ in all cases but one: If none of the three first sub-intervals is part of the decomposition, the best we can do is to approximate $I(c_j)$ by, for instance, $I_9(c_j)$, resulting in a total change of $1$ for the interval.

The CVP instance has $k=20s=10t$ allowed segments. The first $3t$ segments correspond to pairs $(y_i,c_j)$ where $y_i \in \{x_i,\bar{x}_i\}$ is a literal and $c_j$ is a clause involving $y_i$. Consider such a pair $(y_i,c_j)$ and assume that $y_i$ is the $\alpha$-th literal of $c_j$, and $c_j$ is the $\beta$-th clause containing $y_i$ (thus $\alpha, \beta \in \{1,2,3\}$). The segment associated to the pair $(y_i,c_j)$ is $S(y_i,c_j) := (I_\beta(y_i),I_\alpha(c_j))$. The last $7t$ segments are of the form $S_\gamma(c_j) := (\varnothing,I_\gamma(c_j))$ where $c_j$ is a clause and $\gamma \in \{4,\ldots,10\}$. We denote the resulting set of segments by $\mathcal{S} = \{S_1,\ldots,S_k\}$.
This concludes the description of the reduction. Note that the reduction is clearly polynomial. 

Now suppose we have a truth assignment that satisfies $\sigma$ of the clauses of $\Phi$. Then we can find an approximate decomposition of $A$ with a total change of $t-\sigma$ by summing all $3s$ segments of the form $S(y_i,c_j)$ where $y_i = x_i$ if $x_i$ is set to true, and $y_i = \bar{x}_i$ if $x_i$ is set to false, and a subset of the segments $S_\gamma(c_j)$ for all clauses $c_j$ of the formula.

Conversely, assume that we have an approximate decomposition $B := \sum_{j=1}^k u_j S_j$ of $A$ with a total change of $\TC := \|A-B\|_1$. Because $A$ is binary, we may assume that $\mathbf{u}$ is also binary. We say that a variable $x_i$ is {\sl coherent} (w.r.t.\ $B$) if the deviation $\TC(x_i)$ for the interval $I(x_i)$ is zero. The variable is said to be {\sl incoherent} otherwise. Similarly, we say that a clause $c_j$ is {\sl satisfied} (again, w.r.t.\ $B$) if the deviation $\TC(c_j)$ for the interval $I(c_j)$ is zero, and {\sl unsatisfied} otherwise.

We can modify the approximation $B = \sum_{j=1}^k u_j S_j$ in such a way to make all variables coherent, without increasing $\TC$, for the following reasons.

First, while there exist numbers $\alpha$, $i$, $j$ and $j'$ such that $S(x_i,c_j)$ and $S(\bar{x}_i,c_{j'})$ are both used in the decomposition and $c_j$, $c_{j'}$ are the $\alpha$-th clauses respectively containing $x_i$ and $\bar{x}_i$, we can remove one of these two segments from the decomposition and change the segments of the form $S_\gamma(c_j)$ or $S_\gamma(c_{j'})$ that are used, in such a way that $\TC$ does not increase. More precisely, $\TC(x_i)$ decreases by at least one, and $\TC(c_j)$ or $\TC(c_{j'})$, but not both, increases by at most one. Thus we can assume that, for every indices $\alpha$, $i$, $j$ and $j'$ such that $c_j$ is the $\alpha$-th clause containing $x_i$ and $c_{j'}$ is the $\alpha$-th clause containing $\bar{x}_i$, at most one of the two segments $S(x_i,c_j)$ and $S(\bar{x}_i,c_{j'})$ is used in the decomposition. Similarly, we can assume that, for every indices $\alpha$,  $i$, $j$ and $j'$ such that $c_j$ and $c_{j'}$ are the $\alpha$-th clauses respectively containing $x_i$ and $\bar{x}_i$, at least one of the two segments $S(x_i,c_j)$ and $S(\bar{x}_i,c_{j'})$ is used in the decomposition. All in all, we have that exactly one of the two segments $S(x_i,c_j)$ and $S(\bar{x}_i,c_{j'})$ is used.

Second, while some variable $x_i$ remains incoherent, we can replace the segment of the form $S(y_i,c_j)$, where $y_i \in \{x_i,\bar{x}_i\}$, with a segment of the form $S(\bar{y}_i,c_{j'})$ and change some segments of the form $S_\gamma(c_j)$ or $S_\gamma(c_{j'})$ ensuring that $\TC$ does not increase (again, $\TC(x_i)$ decreases by at least one and either $\TC(c_j)$ or $\TC(c_{j'})$ increases by at most one). Therefore, we can assume that all variables are coherent.

Now, by interpreting the relevant part of the decomposition of $B$ as a truth assignment, we obtain truth values for the variables of $\Phi$ satisfying at least $t - \TC$ of its clauses.

Letting $\OPT(\Phi)$ and $\OPT(A,\mathcal{S},\infty)$ denote the maximum number of clauses of $\Phi$ that can be simultaneously satisfied and the total change of an optimal solution of the CVP with $\bd = \infty$, we have $\OPT(\Phi) = t - \OPT(A,\mathcal{S},\infty)$.

Using Lemma \ref{lem:reduction} we know that there exists some $\delta \in (0,1)$ such that it is NP-hard to distinguish, among all 3SAT-6 instances $\Phi$ with $t$ clauses, instances such that $\OPT(\Phi) = t$ from instances such that $\OPT(\Phi) < \delta\,t$. Let $\varepsilon = (1-\delta)/5$. Using the above reduction and the same ideas as in the proof of Theorem \ref{hard_approx}, the theorem follows.
\end{proof}


\section{Incorporating the beam-on time into the objective function.}

In this section we generalize the results of this paper in the case where we do not only want to minimize the total change, but a combination of the total change and the sum of the coefficients $u_j$ for $j\in[k]$. More precisely, we replace the original objective function $\|\mathbf{a} - \mathbf{b}\|_1$ by
$$
\mu \cdot \|\mathbf{a} - \mathbf{b}\|_1 + \nu \cdot \sum_{j=1}^k u_j,
$$
where $\mu$ and $\nu$ are arbitrary nonnegative importance factors. Throughout this section, we study the CVP under this objective function. The resulting problem is denoted {\sl CVP-BOT}.

Let us recall that in the IMRT context the generators from $\mathcal{G}$ represent segments that can be generated by the MLC. The coefficient $u_j$ associated to the segment $\mathbf{g}_j$ for $j \in [k]$ gives the total time during which the patient is irradiated with the leaves of the MLC in a certain position. Hence, the sum of the coefficients exactly corresponds to the total time during which the patient is irradiated ({\sl beam-on time}). In order to avoid overdosage in the healthy tissues due to unavoidable diffusion effects as much as possible, it is desirable to take the beam-on time into account in the objective function.

Here, we observe that the main results of the previous sections still hold with the new objective function.

First, for the hardness results, this is obvious because taking $\mu = 1$ and $\nu = 0$ gives back the original objective function.

Second, for showing that CVP-BOT is polynomial when matrix $G$ defined by the generators is totally unimodular, we use the following LP relaxation:
\begin{equationarray*}{r@{\quad}rcl@{\qquad}l}
\textrm{(LP')} &\multicolumn{3}{@{}l}{\min \mu \cdot \sum_{i=1}^d (\alpha_i + \beta_i) + \nu \cdot \sum_{j=1}^k u_j}\\
\mbox{s.t.} &\sum_{j=1}^{k}u_jg_{ij} - \alpha_i + \beta_i &=&a_i& \forall i\in\left[d\right]\\
&\alpha_i &\geqslant& 0 & \forall i\in\left[d\right]\\[1ex]
&\beta_i &\geqslant& 0 & \forall i\in\left[d\right]\\[1ex]
&\alpha_i &\leqslant& \bd & \forall i\in\left[d\right]\\[1ex]
&\beta_i &\leqslant& \bd & \forall i\in\left[d\right]\\[1ex]
&u_j&\geqslant& 0 & \forall j\in\left[k\right].
\end{equationarray*}
Furthermore, if the columns of $G$ satisfy the consecutive ones property, we can still give a direct reduction to the minimum cost flow problem. Indeed, it suffices to redefine the cost of the arcs of $D$ by letting the cost of arcs of the form $(j,j+1)$ or $(j+1,j)$ (for $j \in [d]$) be $\mu$, and the costs of the other arcs be $\nu$.

Finally, we can also find an $\Big(O(\sqrt{d \ln d}),O(d\,\sqrt{d \ln d})\Big)$-approximation algorithm for CVP-BOT, by using an extension of the randomized rounding technique due to Srivinasan \cite{Sri01}, and its recent derandomization by Doerr and Wahlstr\"om \cite{Doerr09}.

Consider an instance $(H,\mathbf{x})$ of the lattice approximation problem. Assume that $\sum_{j=1}^d x_j \in \mathbb{Z}_+$. We wish to round $\mathbf{x}$ to a binary vector $\mathbf{y}$ such that $\sum_{j=1}^d x_j = \sum_{j=1}^d y_j$ and $\|H(\mathbf{x}-\mathbf{y})\|_\infty = O(\sqrt{d \ln d}\,)$. Srivinasan \cite{Sri01} obtained a randomized polynomial-time algorithm achieving this with high probability. A recent result of Doerr and Wahlstr\"om \cite{Doerr09} implies the following theorem.

\begin{thm}[\cite{Doerr09}] \label{thm:Doerr09}  Let $H \in \{0,1\}^{d \times d}$ and $\mathbf{x} \in \mathbb{Q}^d \cap [0,1]^d$ such that $\sum_{j=1}^d x_j\in\mathbb{Z}_+$. Then, a vector $\mathbf{y}$ can be computed in time $O(d^2)$ such that $\sum_{j=1}^d y_j=\sum_{j=1}^d x_j$ and
\[
\left\|H (\mathbf{x}-\mathbf{y})\right\|_{\infty}\leqslant (e-1)\sqrt{d \ln d}.
\]
\end{thm}

Let again $(\boldsymbol{\alpha^*},\boldsymbol{\beta^*},\mathbf{u}^*)$ denote any extremal optimal solution of (LP'). Recall that at most $d$ of the $k$ components of $\mathbf{u^*}$ are nonzero. Without loss of generality, we can assume that $u^*_j = 0$ for $j > d$.

Now, define $H$ and $\mathbf{x}$ as previously. Because it might be the case that $\sum_{j=1}^d x_j \notin \mathbb{Z}_+$, we turn $H$ and $\mathbf{x}$ respectively into a $(d+1) \times (d+1)$ matrix and a $(d+1) \times 1$ vector by letting $x_{d+1} := \left\lceil \sum_{j=1}^d x_j\right\rceil- \sum_{j=1}^d x_j$ and $h_{d+1,j} = h_{i,d+1} := 0$ for all $i, j \in [d+1]$.

By Theorem \ref{thm:Doerr09}, one can find in $O(d^2)$ time a vector $\mathbf{y} \in \{0,1\}^{d+1}$ such that $\sum_{j=1}^{d+1} y_j = \sum_{j=1}^{d+1} x_j$ and $\|H (\mathbf{x}-\mathbf{y})\|_{\infty} \leqslant (e-1)\sqrt{(d+1) \ln (d+1)} = O(\sqrt{d \ln d}\,)$. We then let $\tilde{u}_j = \lfloor u^*_j \rfloor + y_j$ for $j \in [d]$ and $\tilde{u}_j = 0$ for $j \in [k] \setminus [d]$. The corresponding approximation of $\mathbf{a}$ is $\mathbf{\tilde{b}} := G\mathbf{\tilde{u}}$. Notice that the beam-on-time $\sum_{j=1}^d u_j^*$ will be rounded to $\lfloor \sum_{j=1}^d u_j^* \rfloor$ if $y_{d+1}=1$ and to $\lceil \sum_{j=1}^d u_j^* \rceil$ if $y_{d+1}=0$. Using similar arguments as those used in the proof of Theorem \ref{thm:rounding1}, we see that Theorem \ref{thm:Doerr09} leads to a polynomial-time $\Big(O(\sqrt{d \ln d}),O(d\,\sqrt{d \ln d})\Big)$-approximation algorithm for CVP-BOT.


\section{Conclusion} \label{sec:conclusion}

Here are further questions we leave open for future work:

\begin{itemize}
\item Our results confirm that it is worth to solve the natural LP relaxation of the problem (see page \pageref{sec:TU}). This can be done efficiently when the generators are explicitly given. However, in practical applications, the generators are implicitly given. Improving our understanding of when the LP relaxation can still be solved in polynomial-time is a first interesting open question.

\item Our second question concerns the tightness of the (in)approximability results developed here, in particular for the case $\bd = \infty$. We prove that there is no polynomial-time approximation algorithm with an additive approximation guarantee of $(\ln 2-\varepsilon)\,d$, unless P $=$ NP. On the other hand, we give a randomized approximation algorithm with an $O(d\sqrt{d}\,)$ additive approximation guarantee. What is the true approximability threshold of the CVP (restricted to instances where $\bd = \infty$)?

\item Our third question is more algorithmic and concerns the application of the CVP to IMRT. There exists a simple, direct algorithm for checking whether an intensity matrix $A$ can be decomposed exactly under the minimum separation constraint or not \cite{Kal06,Kam03a}. Is there a simple, direct algorithm for approximately decomposing an intensity matrix under this constraint as well?
\end{itemize}


\section{Acknowledgments}

We thank \c{C}i\u{g}dem G\"uler and Horst Hamacher from the University of Kaiserslautern for several discussions in the early stage of this research project. We also thank Maude Gathy, Guy Louchard and Yvik Swan from Universit\'e Libre de Bruxelles, and Konrad Engel and Thomas Kalinowski from Universit\"at Rostock for discussions.

We also thank three anonymous referees for their useful comments.

\medskip

\noindent\textbf{Note added during the revision of this manuscript.} A recent breakthrough result due to Bansal \cite{Bansal10} might lead to a randomized $(O(\sqrt{d}),O(d\sqrt{d}))$-approximation algorithm for the CVP, hence improving both Theorem \ref{thm:rounding1} and Theorem \ref{thm:rounding2}.


\newpage

\section*{Appendix}

\noindent \textbf{Lemma.}
\textit{Let $q$ be a positive integer and let $X_1,X_2,\dots,X_q$ be $q$ independent random variables such that, for all $j\in [q]$, $P[X_j = 1-p_j] = p_j$ and $P[X_j = -p_j] = 1-p_j$. Then
$$
E\Big[|X_1 + X_2 + \cdots + X_q|\Big] \leqslant \sqrt{\frac{\ln 2}{2}}\;\sqrt{q}.
$$}

\begin{proof}
Let $X := X_1 + X_2 + \cdots + X_q$. By Lemmas A.1.6 and A.1.8 from Alon and Spencer \cite{alon}, we have
$$
E[e^{\lambda X}] \leqslant \Big(p\,e^{(1-p)\lambda} + (1-p)\,e^{-\lambda p}\Big)^q
\leqslant e^{\frac{\lambda^2}{8}\,q},
$$
for every $\lambda \in \mathbb{R}$ and $p:=\frac{1}{q} \sum_{j=1}^q p_j$. Hence,
$$
e^{E[\lambda |X|]} \leqslant E[e^{\lambda |X|}]
\leqslant E[e^{\lambda X}] + E[e^{-\lambda X}]
\leqslant 2e^{\frac{\lambda^2}{8}\,q},
$$
where the first inequality follows from Jensen's inequality (and the fact that the exponential is convex) and the second inequality follows from the easy inequality $e^{|x|} \leqslant e^{x} + e^{-x}$. Taking logarithms and letting $\lambda = \sqrt{\frac{8 \ln 2}{q}}$, we get the result.
\end{proof}
\end{document}